\documentclass[11pt]{article}
\usepackage[utf8]{inputenc}
\usepackage{hyperref}
\usepackage{fullpage}
\usepackage{amssymb}
\usepackage{mathtools}
\usepackage{amsthm}
\usepackage{eucal}
\usepackage{dsfont}
\usepackage[T1]{fontenc}
\usepackage{lmodern}
\usepackage[stretch=10,shrink=10]{microtype}
\usepackage[capitalise]{cleveref}
\usepackage{color,soul}
\usepackage[framemethod=tikz]{mdframed}
\usepackage{lipsum}

\allowdisplaybreaks[1]

\theoremstyle{plain}
\newtheorem{theorem}{Theorem}[section]
\newtheorem{lemma}[theorem]{Lemma}

\newtheorem{cor}[theorem]{Corollary}

\theoremstyle{definition}
\newtheorem{definition}[theorem]{Definition}

\newtheorem{claim}[theorem]{Claim}

\newtheorem{fact}[theorem]{Fact}

\newcommand {\eps} {\varepsilon}
\newcommand {\br} [1] {\ensuremath{ \left( #1 \right) }}

\newcommand {\minusspace} {\: \! \!}
\newcommand {\smallspace} {\: \!}
\newcommand {\fn} [2] {\ensuremath{ #1 \minusspace \br{ #2 } }}

\newcommand {\ball} [2] {\fn{\mathcal{B}^{#1}}{#2}}
\newcommand {\defeq} {\ensuremath{ \stackrel{\mathrm{def}}{=} }}
\newcommand {\mutinf} [2] {\fn{\mathrm{I}}{#1 \smallspace : \smallspace #2}}
\newcommand {\imax} [2] {\fn{\mathrm{I}_{\max}}{#1 \smallspace : \smallspace #2}}
\newcommand {\imaxeps} [2] {\fn{\mathrm{I}^{\varepsilon}_{\max}}{#1 \smallspace : \smallspace #2}}
\newcommand {\imaxdelta} [2] {\fn{\mathrm{I}^{\delta}_{\max}}{#1 \smallspace : \smallspace #2}}
\newcommand {\condmutinf} [3] {\mutinf{#1}{#2 \smallspace \middle\vert \smallspace #3}}

\newcommand {\norm} [1] {\ensuremath{ \left\| #1 \right\| }}
\newcommand {\normsub} [2] {\ensuremath{ \norm{#1}_{#2} }}
\newcommand {\onenorm} [1] {\normsub{#1}{1}}
\newcommand {\ent} [1] {\fn{\mathrm{H}}{#1}}
\newcommand {\relent} [2] {\fn{\mathrm{D}}{#1 \middle\| #2}}
\newcommand {\dmax} [2] {\fn{\mathrm{D}_{\max}}{#1 \middle\| #2}}
\newcommand {\hmin} [2] {\fn{\mathrm{H}_{\min}}{#1 \middle | #2}}
\newcommand {\hmineps} [2] {\fn{\mathrm{H}^{\varepsilon}_{\min}}{#1 \middle | #2}}
\newcommand {\hmindelta} [2] {\fn{\mathrm{H}^{\delta}_{\min}}{#1 \middle | #2}}
\newcommand {\hmax} [2] {\fn{\mathrm{H}_{\max}}{#1 \middle | #2}}
\newcommand {\hmaxeps} [2] {\fn{\mathrm{H}^{\varepsilon}_{\max}}{#1 \middle | #2}}

\newcommand {\email} [1] {\href{mailto:#1}{\texttt{#1}}}

\newcommand {\bra} [1] {\ensuremath{ \left\langle #1 \right| }}
\newcommand {\ket} [1] {\ensuremath{ \left| #1 \right\rangle }}
\newcommand {\ketbratwo} [2] {\ensuremath{ \left| #1 \middle\rangle \middle\langle #2 \right| }}
\newcommand {\ketbra} [1] {\ketbratwo{#1}{#1}}
\newcommand {\cspace} [1] {\ensuremath{\mathnormal{#1}}}

\newcommand {\Tr} {\ensuremath{ \mathrm{Tr} }}
\newcommand {\partrace} [2] {\fn{\Tr_{#1}}{#2}}
\newcommand {\id} {\ensuremath{\mathds{1}}}

\newcommand {\suppress}[1]{}

\newcommand {\set} [1] {\ensuremath{ \left\lbrace #1 \right\rbrace }}
\newcommand {\reg} [1] {\ensuremath{ \mathnormal{#1} }}

\def\B{\CMcal{B}}

\def\M{\CMcal{M}}

\def\G{\CMcal{G}}
\def\P{\mathrm{P}}

\def\T{\CMcal{T}}

\def\O{\CMcal{O}}

\def\H{\CMcal{H}}

\def\F{\mathrm{F}}

\def\inf{\mathrm{inf}}
\def\max{\mathrm{max}}

\def\min{\mathrm{min}}

\def\E{\mathcal{E}}
\def\M{\CMcal{M}}

\newcommand {\mytitle} {A lower bound on expected communication cost of quantum state redistribution}
\newcommand {\Anurag}  {Anurag Anshu}

\newcommand {\CQT} {Centre for Quantum Technologies}
\newcommand {\NUS} {National University of Singapore}

\newcommand {\authorblock} [3] {
	\begin{minipage}[t]{0.3\linewidth}
		\centering
		{#1}\\[0.8ex]
		{\footnotesize {#2}\\[-0.7ex]
		\email{#3}}
	\end{minipage}\vspace{1ex}
}

\hypersetup{
	pdfstartview={FitH},
	pdfdisplaydoctitle={true},
	breaklinks={true},
	bookmarksopen={true},
	bookmarksnumbered={false},
	pdftitle={\mytitle},
	pdfauthor={\Anurag}
}

\begin{document}

\title{\textbf{\mytitle}\\[2ex]}

\author{
    \authorblock{\Anurag}{\CQT, \NUS}{a0109169@u.nus.edu}
}

\maketitle

\begin{abstract}
We show a lower bound on expected communication cost of interactive entanglement assisted quantum state redistribution protocols and a slightly better lower bound for its special case, quantum state transfer. Our bound implies that the expected communication cost of interactive protocols is not significantly better than worst case communication cost, in terms of scaling of error. Furthermore, the bound is independent of the number of rounds. This is in contrast with the classical case, where protocols with expected communication cost significantly better than worst case communication cost are known.
\end{abstract}
\section{Introduction}
\label{sec:intro}
A fundamental task in quantum information theory is that of quantum state redistribution (various quantities appearing in this section have been described in Section \ref{sec:preliminaries}): 

\bigskip

\textbf{Quantum state-redistribution} : A pure state $\Psi_{RBCA}$ is shared between Alice (A,C), Bob(B) and Referee(R). For a given $\eps > 0$, which we shall henceforth identify as `error', Alice needs to transfer the system $C$ to Bob, such that the final state $\Psi'_{RBC_0A}$ (where register $C_0\equiv C$ is with Bob), satisfies $\P(\Psi'_{RBC_0A},\Psi_{RBC_0A})\leq \eps$. Here, $\P(.,.)$ is the purified distance.

\bigskip

This task has been well studied in literature in asymptotic setting (\cite{Devatakyard,oppenheim08,YeBW08,YardD09}) giving an operational interpretation to the quantum conditional mutual information. Recent results have obtained one-shot versions of this task (~\cite{Oppenheim14,Berta14,jain14}), with application to bounded-round entanglement assisted quantum communication complexity (\cite{Dave14}). 

The following upper bound has been obtained in \cite{Dave14}, developing upon the work in \cite{Berta14}, on worst case communication cost of quantum state redistribution, with error $\eps$:  $$\frac{50\cdot\condmutinf{R}{C}{B}_{\Psi_{RABC}}}{2\eps^2} + \frac{100}{\eps^2} + 15.$$ An important application of this bound is a direct sum theorem for communication cost of bounded-round entanglement assisted quantum communication complexity, which is the main result of \cite{Dave14}:

\begin{theorem}[Touchette \cite{Dave14}, Theorem $3$]
\label{dave}
 Let $C$ be the quantum communication complexity of the best entanglement assisted protocol for computing a relation $f$ with error $\rho$ on inputs drawn from a distribution $\mu$. Then any $r$ round protocol computing $f^{\otimes n}$ on the distribution $\mu^{\otimes n}$ with error $\rho-\eps$ must involve at least $\Omega(n((\frac{\eps}{r})^2\cdot C - r))$ quantum communication.
\end{theorem}

 Direct sum results for single-round entanglement assisted quantum communication complexity had earlier been obtained in ~\cite{Jain:2005,Jain:2008,AnshuJMSY2014}. 

A special case of quantum state redistribution is \textbf{quantum state merging}, in which the register $A$ is absent. It was introduced in \cite{horodecki07} as a quantum counterpart to the classical Slepian-Wolf protocol \cite{slepianwolf}. A one-shot quantum state merging was introduced by Berta \cite{Berta09}. A one-shot version of classical Slepian-Wolf protocol was obtained by Braverman and Rao\cite{bravermanrao11}, in the form of the following task:

\bigskip

Alice is given a probability distribution $P$, Bob is given a probability distribution $Q$.  Bob must output a distribution $P'$, with the property that $\|P-P'\|_1\leq \eps$. 

\bigskip

They exhibited an interactive communication protocol achieving this task with \textit{expected communication cost} $$\relent{P}{Q} + \mathcal{O}(\sqrt{\relent{P}{Q}})+2\log(\frac{1}{\eps}).$$ Considering expected communication cost, instead of worst case communication cost, allowed them to obtain the following direct sum result for bounded round classical communication complexity: 

\begin{theorem}[Braverman and Rao \cite{bravermanrao11}, Corollary 2.5]
\label{brarao}
Let $C$ be the communication complexity of the best protocol for computing a relation $f$ with error $\rho$ on inputs drawn from a distribution $\mu$. Then any $r$ round protocol computing $f^{\otimes n}$ on the distribution $\mu^{\otimes n}$ with error $\rho-\eps$ must involve at least $\Omega(n(C - r\cdot\log(\frac{1}{\eps}) - O(\sqrt{C\cdot r})))$ communication.
\end{theorem}

This result has better dependence on number of rounds $r$ in comparison to theorem \ref{dave}. Thus, in order to obtain a stronger direct sum result for bounded-round quantum communication complexity, a possible approach would be to bound the expected communication cost of quantum state redistribution by $\approx \condmutinf{R}{C}{B}_{\Psi_{RABC}}+\mathcal{O}(\log(\frac{1}{\eps}))$. 

A special case of quantum state merging is \textbf{quantum state transfer}, in which register $B$ is trivial. Asymptotic version of quantum state transfer is the Schumacher compression \cite{Schumacher95}. In the corresponding classical setting, when $\Psi_{RA}$ is a classical probability distribution, Alice can send register $A$ to Bob with expected communication cost $S(\Psi_A) + \mathcal{O}(1)$, using a one-way protocol based on Huffman coding \cite{CoverT91}. In fact, one can make the error arbitrarily small, at the cost of arbitrarily large worst case communication. 

\subsection*{Our results}

In this work, we show that expected communication cost for entanglement assisted quantum protocols (which we formally define in section \ref{sec:cohtrans}) is not significantly better than the worst case communication cost. Our main theorem is the following.

\begin{theorem}
\label{thm:main}
Fix a $p<1$ and an $\eps \in [0,(\frac{1}{70})^{\frac{4}{1-p}}]$. There exists a pure state $\Psi_{RBCA}$ (that depends on $\eps$) such that, any interactive entanglement assisted communication protocol for its quantum state redistribution with error $\eps$ requires expected communication cost at least $\condmutinf{R}{C}{B}_{\Psi}\cdot(\frac{1}{\eps})^{p}$.
\end{theorem}

For quantum state transfer, we obtain a similar result with slightly better constants.

\begin{theorem}
\label{thm:main2}
Fix a $p<1$ and any $\eps \in [0, (\frac{1}{2})^{\frac{15}{1-p}}]$. There exists a pure state $\Psi_{RC}$ (that depends on $\eps$) such that, any interactive entanglement assisted communication protocol for its quantum state transfer with error $\eps$ requires expected communication cost at least $S(\Psi_R)\cdot(\frac{1}{\eps})^p$.
\end{theorem}

Notice that theorem \ref{thm:main2} does imply theorem \ref{thm:main}, as quantum state transfer is a special case of quantum state redistribution. But the quantum state $\Psi_{RBCA}$ that we consider in theorem \ref{thm:main} has all registers $R,A,B,C$ non-trivial and correlated with each other. Thus, a quantum state redistribution of $\Psi_{RBCA}$ cannot be reduced to the sub-cases of quantum state merging or quantum state transfer by any local operation, giving robustness to the bound. 
\subsection*{Our technique and organization}

We discuss our technique for the case of quantum state transfer. For some $\beta>1$, we choose the pure state $\Psi_{RC}$ in such a way that its smallest eigenvalue is $\frac{1}{d\beta}$ and entropy of $\Psi_R$ is at most $\frac{2\log(d)}{\beta}$ ($d$ being dimension of register $R$). Let $\omega_{RC}$ be a maximally entangled state defined as $\ket{\omega}_{RC} = \frac{\Psi_R^{-\frac{1}{2}}}{\sqrt{d}}\ket{\Psi}_{RC}$.  For any interactive protocol $\mathcal{P}$ for quantum state transfer of $\Psi_{RC}$ with error $\eps$ and expected communication cost $C$, we obtain an expression that serves as a \textit{transcript} of the protocol, encoding the unitaries applies by Alice and Bob and the probabilities of measurement outcomes (Corollary \ref{cohequation}). This expression is obtained by employing a technique of \textit{convex-split}, introduced in \cite{jain14} for one-way quantum state redistribution protocols. Then, crucially relying on the fact that $\Psi_{RC}$ is a pure state, we construct a new interactive protocol $\mathcal{P'}$ which achieves quantum state transfer of the state $\omega_{RC}$ with error $\sqrt{\beta\eps}+\sqrt{\mu}$ (for any $\mu<1$) and worst case quantum communication cost at most $\frac{C}{\mu}$. Suitably choosing the parameters $\eps,\beta$ and $\mu$ and using known lower bound on worst case communication cost for state transfer of $\omega_{RC}$, we obtain the desired result. Same technique also extends to quantum state redistribution. Details appear in section \ref{sec:lowerbound}.

In section \ref{sec:preliminaries} we present some notions and facts that are needed for our proofs. In section \ref{sec:cohtrans} we give a description of interactive protocols for quantum state redistribution and obtain the aforementioned expression that serves as a \textit{transcript} of a given protocol. Section \ref{sec:lowerbound} is devoted to refinement of this expression and proof of main theorem. We present some discussion related to our approach and conclude in Section \ref{sec:conclusion}.
    
\section{Preliminaries}
\label{sec:preliminaries}
In this section we present some notations, definitions, facts and lemmas that we will use in our proofs.
\subsection*{Information theory}

For a natural number $n$, let $[n]$ represent the set $\{1,2, \dots, n\}$. For a set $S$, let $|S|$ be the size of $S$. A \textit{tuple} is a finite collection of positive integers, such as $(i_1,i_2\ldots i_r)$ for some finite $r$. We let $\log$ represent logarithm to the base $2$ and $\ln$ represent logarithm to the base $\mathrm{e}$. The $\ell_1$ norm of an operator $X$ is $\onenorm{X}\defeq\Tr\sqrt{X^{\dag}X}$ and $\ell_2$ norm is $\norm{X}_2\defeq\sqrt{\Tr XX^{\dag}}$. A quantum state (or just a state) is a positive semi-definite matrix with trace equal to $1$. It is called {\em pure} if and only if the rank is $1$. Let $\ket{\psi}$ be a unit vector.  We use $\psi$ to represent the state
and also the density matrix  $\ketbra{\psi}$, associated with $\ket{\psi}$. 

A sub-normalized state is a positive semidefinite matrix with trace less than or equal to $1$. A {\em quantum register} $A$ is associated with some Hilbert space $\H_A$. Define $|A| \defeq \dim(\H_A)$. We denote by $\mathcal{D}(A)$, the set of quantum states in the Hilbert space $\H_A$ and by $\mathcal{D}_{\leq}(A)$, the set of all subnormalized states on register $A$. State $\rho$ with subscript $A$ indicates $\rho_A \in \mathcal{D}(A)$.

 For two quantum states $\rho$ and $\sigma$, $\rho\otimes\sigma$ represents the tensor product (Kronecker product) of $\rho$ and $\sigma$.  Composition of two registers $A$ and $B$, denoted $AB$, is associated with Hilbert space $\H_A \otimes \H_B$. If two registers $A,B$ are associated with the same Hilbert space, we shall denote it by $A\equiv B$. Let $\rho_{AB}$ be a bipartite quantum state in registers $AB$.  We define
\[ \rho_{\reg{B}} \defeq \partrace{\reg{A}}{\rho_{AB}}
\defeq \sum_i (\bra{i} \otimes \id_{\cspace{B}})
\rho_{AB} (\ket{i} \otimes \id_{\cspace{B}}) , \]
where $\set{\ket{i}}_i$ is an orthonormal basis for the Hilbert space $\cspace{A}$
and $\id_{\cspace{B}}$ is the identity matrix in space $\cspace{B}$.
The state $\rho_B$ is referred to as the marginal state of $\rho_{AB}$ in register $B$. Unless otherwise stated, a missing register from subscript in a state will represent partial trace over that register. A quantum map $\E: A\rightarrow B$ is a completely positive and trace preserving (CPTP) linear map (mapping states from $\mathcal{D}(A)$ to states in $\mathcal{D}(B)$). A completely positive and trace non-increasing linear map $\tilde{\E}:A\rightarrow B$ maps quantum states to sub-normalised states. The identity operator in Hilbert space $\H_A$ (and associated register $A$) is denoted $I_A$.  A {\em unitary} operator $U_A:\H_A \rightarrow \H_A$ is such that $U_A^{\dagger}U_A = U_A U_A^{\dagger} = I_A$. An {\em isometry}  $V:\H_A \rightarrow \H_B$ is such that $V^{\dagger}V = I_A$ and $VV^{\dagger} = I_B$. The set of all unitary operations on register $A$ is  denoted by $\mathcal{U}(A)$. 

\begin{definition}
We shall consider the following information theoretic quantities. Let $\varepsilon \geq 0$. 
\begin{enumerate}
\item {\bf generalized fidelity} For $\rho,\sigma \in \mathcal{D}_{\leq}(A)$, $$\F(\rho,\sigma)\defeq \onenorm{\sqrt{\rho}\sqrt{\sigma}} + \sqrt{(1-\Tr(\rho))(1-\Tr(\sigma))}.$$ 
\item {\bf purified distance} For $\rho,\sigma \in \mathcal{D}_{\leq}(A)$, $$\P(\rho,\sigma) = \sqrt{1-\F^2(\rho,\sigma)}.$$
\item {\bf $\varepsilon$-ball} For $\rho_A\in \mathcal{D}(A)$, $$\ball{\eps}{\rho_A} \defeq \{\rho'_A\in \mathcal{D}(A)|~\P(\rho_A,\rho'_A) \leq \varepsilon\}. $$ 
\item {\bf entropy} For $\rho_A\in \mathcal{D}(A)$, $$\ent{A}_{\rho} \defeq - \Tr(\rho_A\log\rho_A) .$$ 
\item {\bf relative entropy} For $\rho_A,\sigma_A\in \mathcal{D}(A)$, $$\relent{\rho_A}{\sigma_A} \defeq \Tr(\rho_A\log\rho_A) - \Tr(\rho_A\log\sigma_A) .$$ 
\item {\bf max-relative entropy} For $\rho_A,\sigma_A\in \mathcal{D}(A)$, $$ \dmax{\rho_A}{\sigma_A}  \defeq  \inf \{ \lambda \in \mathbb{R} : 2^{\lambda} \sigma_A \geq \rho_A \}  .$$ 
\item {\bf mutual information}  For $\rho_{AB}\in \mathcal{D}(AB)$,$$\mutinf{A}{B}_{\rho}  \defeq \relent{\rho_{AB}}{\rho_A\otimes\rho_B}= \ent{A}_{\rho} + \ent{B}_{\rho} - \ent{AB}_{\rho}.$$
\item {\bf conditional mutual information} For $\rho_{ABC}\in \mathcal{D}(ABC)$, $$\condmutinf{A}{B}{C}_{\rho}  \defeq \mutinf{A}{BC}_{\rho}  -  \mutinf{A}{C}_{\rho} = \mutinf{B}{AC}_{\rho}  -  \mutinf{B}{C}_{\rho} .$$
\item {\bf max-information} For $\rho_{AB}\in \mathcal{D}(AB)$, $$ \imax{A}{B}_{\rho} \defeq   \inf_{\sigma_{B}\in \mathcal{D}(B)}\dmax{\rho_{AB}}{\rho_{A}\otimes\sigma_{B}} .$$
\item {\bf smooth max-information} For $\rho_{AB}\in \mathcal{D}(AB)$, $$\imaxeps{A}{B}_{\rho} \defeq \inf_{\rho'\in \ball{\eps}{\rho}} \imax{A}{B}_{\rho'} .$$	
\item {\bf conditional min-entropy} For $\rho_{AB}\in \mathcal{D}(AB)$, $$ \hmin{A}{B}_{\rho} \defeq  - \inf_{\sigma_B\in \mathcal{D}(B)}\dmax{\rho_{AB}}{I_{A}\otimes\sigma_{B}} .$$  	
\item {\bf conditional max-entropy} For $\rho_{AB}\in \mathcal{D}(AB)$, $$\hmax{A}{B}_{\rho_{AB}} \defeq - \hmin{A}{R}_{\rho_{AR}}, $$
where $\rho_{ABR}$ is a  purification of $\rho_{AB}$ for some system $R$. 
\item {\bf smooth conditional min-entropy} For $\rho_{AB}\in \mathcal{D}(AB)$, $$\hmineps{A}{B}_{\rho} \defeq   \sup_{\rho^{'} \in \ball{\eps}{\rho}} \hmin{A}{B}_{\rho^{'}} .$$  	
\item {\bf smooth conditional max-entropy} For $\rho_{AB}\in \mathcal{D}(AB)$, $$ \hmaxeps{A}{B}_{\rho} \defeq \inf_{\rho^{'} \in \ball{\eps}{\rho}} \hmax{A}{B}_{\rho^{'}} .$$ 
\end{enumerate}
\label{def:infquant}
\end{definition}	
We will use the following facts. 
\begin{fact}[Triangle inequality for purified distance,~\cite{Tomamichel12}]
\label{fact:trianglepurified}
For states $\rho^1_A, \rho^2_A, \rho^3_A \in \mathcal{D}(A)$,
$$\P(\rho^1_A,\rho^3_A) \leq \P(\rho^1_A,\rho^2_A)  + \P(\rho^2_A,\rho^3_A) . $$ 
\end{fact}

\begin{fact}[Purified distance and trace distance,~\cite{Tomamichel12}, Proposition 3.3]
\label{fact:purifiedtrace}
For subnormalized states $\rho_1,\rho_2$
$$\frac{1}{2}\|\rho_1-\rho_2\|_1\leq \P(\rho_1,\rho_2) \leq \sqrt{\|\rho_1-\rho_2\|_1}.$$
\end{fact}

\begin{fact}[Uhlmann's theorem][\cite{uhlmann76}]
\label{uhlmann}
Let $\rho_A,\sigma_A\in \mathcal{D}(A)$. Let $\ket{\rho}_{AB}$ be a purification of $\rho_A$ and $\ket{\sigma}_{AC}$ be a purification of $\sigma_A$. There exists an isometry $V: \H_C \rightarrow \H_B$ such that,
 $$\F(\ketbra{\theta}_{AB}, \ketbra{\rho}_{AB}) = \F(\rho_A,\sigma_A) ,$$
 where $\ket{\theta}_{AB} = (I_A \otimes V) \ket{\sigma}_{AC}$.
\end{fact}

\begin{fact}[Monotonicity of quantum operations][\cite{lindblad75, barnum96}, \cite{Tomamichel12}, Theorem 3.4]
\label{fact:monotonequantumoperation}
For states $\rho$, $\sigma$, and quantum operation $\E(\cdot)$, 
$$\onenorm{\E(\rho) - \E(\sigma)} \leq \onenorm{\rho - \sigma} , \P(\rho,\sigma)\leq \P(\E(\rho),\E(\sigma))  \text{ and } \F(\rho,\sigma) \leq \F(\E(\rho),\E(\sigma)) .$$
In particular, for a trace non-increasing completely positive map $\tilde{\E}(\cdot)$, $$\P(\rho,\sigma)\leq \P(\tilde{\E}(\rho),\tilde{\E}(\sigma)).$$

\end{fact}

\begin{fact}[Join concavity of fidelity][\cite{Watrouslecturenote}, Proposition 4.7]\label{fact:fidelityconcave}
Given quantum states $\rho_1,\rho_2\ldots\rho_k,\sigma_1,\sigma_2\ldots\sigma_k \in \mathcal{D}(A)$ and positive numbers $p_1,p_2\ldots p_k$ such that $\sum_ip_i=1$. Then $$\F(\sum_ip_i\rho_i,\sum_ip_i\sigma_i)\geq \sum_ip_i\F(\rho_i,\sigma_i).$$
\end{fact}

\begin{fact}
\label{scalarpurified}
Let $\rho,\sigma \in \mathcal{D}(A)$ be quantum states. Let $\alpha<1$ be a positive real number. If $\P(\alpha\rho,\alpha\sigma)\leq \eps$, then $$\P(\rho,\sigma)\leq \eps\sqrt{\frac{2}{\alpha}}.$$
\end{fact}

\begin{proof}
$\P(\alpha\rho,\alpha\sigma)\leq \eps$ implies $\F(\alpha\rho,\alpha\sigma)\geq \sqrt{1-\eps^2}\geq 1-\eps^2$. But, 
$\F(\alpha\rho,\alpha\sigma)= \alpha\|\sqrt{\rho}\sqrt{\sigma}\|_1+(1-\alpha)$. Thus, $$\F(\rho,\sigma)=\|\sqrt{\rho}\sqrt{\sigma}\|_1 \geq 1-\frac{\eps^2}{\alpha}.$$ Thus, $\P(\rho,\sigma)\leq \sqrt{1-(1-\frac{\eps^2}{\alpha})^2}\leq \sqrt{\frac{2\eps^2}{\alpha}}$.
\end{proof}

\begin{fact}[Fannes inequality][\cite{fannes73}]
\label{fact:fannes}
Given quantum states $\rho_1,\rho_2\in \mathcal{D}(A)$, such that $|A|=d$ and $\P(\rho_1,\rho_2)= \eps \leq \frac{1}{2\mathrm{e}}$, $$|S(\rho_1)-S(\rho_2)|\leq \eps\log(d)+1.$$   
\end{fact}

\begin{fact}[Subadditivity of entropy][\cite{LiebAraki70}]
\label{subadditive}
For a quantum state $\rho_{AB}\in \mathcal{D}(AB)$, $|S(\rho_A)-S(\rho_B)|\leq S(\rho_{AB})\leq S(\rho_A)+S(\rho_B)$.
\end{fact}

\begin{fact}[Concavity of entropy][\cite{Watrouslecturenote}, Theorem 10.9]
\label{entropyconcave}
For quantum states $\rho_1,\rho_2\ldots \rho_n$, and positive real numbers $\lambda_1,\lambda_2\ldots \lambda_n$ satisfying $\sum_i \lambda_i=1$, $$S(\sum_i \lambda_i\rho_i)\geq \sum_i\lambda_iS(\rho_i).$$
\end{fact}

\begin{fact}
\label{informationbound}
For a quantum state $\rho_{ABC}$, it holds that $$\mutinf{A}{C}_{\rho}\leq 2S(\rho_C),$$ $$\condmutinf{A}{C}{B}_{\rho}\leq \mutinf{AB}{C}_{\rho}\leq 2S(\rho_C).$$  
\end{fact}
\begin{proof}
From Fact \ref{subadditive}, $\mutinf{A}{C}_{\rho} = S(\rho_A)+S(\rho_C)-S(\rho_{AC}) \leq 2S(\rho_{C})$. 
\end{proof}

\begin{fact}
\label{fact:imaxhmin}
For a bipartite quantum state $\rho_{AB}$, $\imaxeps{A}{B}_{\rho}\geq -\hmineps{A}{B}_{\rho}$.
\end{fact}

\begin{proof}
 Let $\sigma_B$ be the state achieved in infimum in the definition of $\imax{A}{B}_{\rho}$. Let $\lambda\defeq \imax{A}{B}_{\rho}$. Consider, 
$$\rho_{AB}\leq 2^{\lambda}\rho_A\otimes\sigma_B\leq 2^{\lambda}I_A\otimes\sigma_B.$$ Thus, we have $$-\hmin{A}{B}_{\rho}=\inf_{\sigma'_B\in \mathcal{D}(B)}\dmax{\rho_{AB}}{I_A\otimes\sigma'_B} \leq \dmax{\rho_{AB}}{I_A\otimes\sigma_B} \leq \lambda =  \imax{A}{B}_{\rho}.$$ This gives, $$\inf_{\rho'_{AB}\in\ball{\eps}{\rho_{AB}}}-\hmin{A}{B}_{\rho'} \leq \imaxeps{A}{B}_{\rho}.$$

\end{proof}

\begin{fact}
\label{fact:cqimax}
For a \textit{classical-quantum} state $\rho_{AB}$ of the form $\rho_{AB}=\sum_j p(j)\ketbra{j}_A\otimes \sigma^j_B$, it holds that $\imax{A}{B}_{\rho}\leq \log(|B|)$.
\end{fact}
\begin{proof}
By definition, $\imax{A}{B}_{\rho}\leq \dmax{\rho_{AB}}{\rho_A\otimes\frac{\text{I}_B}{|B|}}$.  Also, $$\rho_{AB}=\sum_j p(j)\ketbra{j}_A\otimes \sigma^j_B \leq |B|\sum_j p(j)\ketbra{j}_A\otimes \frac{\text{I}_B}{|B|} = |B| \rho_A\otimes \frac{\text{I}_B}{|B|}.$$ Thus, the fact follows.
\end{proof}

\begin{fact}
\label{cqmutinf}
For a \textit{classical-quantum} state $\rho_{ABC}=\sum_j p(j)\ketbra{j}_A\otimes \rho^j_{BC}$, it holds that 
$\mutinf{AB}{C}_{\rho}\geq \sum_j p(j)\mutinf{B}{C}_{\rho^j}$ 
\end{fact}
\begin{proof}
Consider, 
\begin{eqnarray*}
\mutinf{AB}{C}_{\rho}&=& S(\rho_{AB}) + S(\rho_C) - S(\rho_{ABC})\\ &=& S(\sum_j p(j)\ketbra{j}_A\otimes\rho^j_B) + S(\sum_j p(j)\rho^j_C) - S(\sum_j p(j)\ketbra{j}_A\otimes\rho^j_{BC}) \\&=& \sum_j p(j)S(\rho^j_B)+S(\sum_j p(j)\rho^j_C) - \sum_j p(j)S(\rho^j_{BC}) \\&\geq&  \sum_j p(j)S(\rho^j_B)+\sum_j p(j)S(\rho^j_C) - \sum_j p(j)S(\rho^j_{BC}) \quad (\text{Fact \ref{entropyconcave}})\\ &=& \sum_j p(j)\mutinf{B}{C}_{\rho^j}
\end{eqnarray*}
\end{proof}

\begin{lemma}
\label{lowentropy}
Fix a $\beta \geq 1$ and an integer $d>1$. There exists a probability distribution $\mu=\{e_1,e_2\ldots e_d\}$, with $e_1\geq e_2 \ldots \geq e_d$, such that $e_d = \frac{1}{d\beta}$ and entropy $S(\mu)\leq 2\frac{\log(d)}{\beta}$
\end{lemma}

\begin{proof}
Set $e_2=e_3=\ldots e_d = \frac{1}{d\beta}$. Then $e_1=1-\frac{d-1}{d\beta}$. Using $x\log(\frac{1}{x})\leq \frac{\log(e)}{e} < 1$ for all $x>0$, we can upper bound the entropy of the distribution as $$\sum_i e_i\log(\frac{1}{e_i}) = (1-\frac{d-1}{d\beta})\log(\frac{1}{1-\frac{d-1}{d\beta}}) + \frac{d-1}{d\beta}\log(d\beta) < 2 + \frac{\log(d)}{\beta}\leq 2\frac{\log(d)}{\beta}.$$ 
\end{proof}

\section{Interactive protocol for quantum state redistribution}
\label{sec:cohtrans}

In this section, we describe general structure of an interactive protocol for quantum state redistribution and its \textit{expected communication cost}.

Let quantum state $\ket{\Psi}_{RBCA}$ be shared between Alice $(A,C)$, Bob $(B)$ and Referee $(R)$. Alice and Bob have access to shared entanglement $\theta_{E_AE_B}$ in registers $E_A$ (with Alice) and $E_B$ (with Bob). Using quantum teleportation, we can assume without loss of generality that Alice and Bob communicate classical messages, which involves performing a projective measurement on registers they respectively hold, and sending the outcome of measurement to other party. This allows for the notion of \textit{expected communication cost}. 

A $r$-round interactive protocol $\mathcal{P}$ (where $r$ is an odd number) with error $\eps$ and expected communication cost $C$ is as follows.
\bigskip
\begin{mdframed}
\bigskip
\quad\textbf{Input:} A quantum state $\ket{\Psi}_{RBCA}$, error parameter $\eps<1$.

\textbf{Shared entanglement:} $\ket{\theta}_{E_AE_B}$.

\begin{itemize}
\item Alice performs a projective measurement $\M=\{M^1_{ACE_A},M^2_{ACE_A}\ldots \}$. Probability of outcome $i_1$ is $p_{i_1}\defeq\Tr(M^{i_1}_{ACE_A}\Psi_{CA}\otimes\theta_{E_A})$. Let $\phi^{i_1}_{RBACE_AE_B}$ be the global normalized quantum state, conditioned on this outcome. She sends message $i_1$ to Bob.
 
\item Upon receiving the message $i_1$ from Alice, Bob performs a projective measurement  $$\M^{i_1}=\{M^{1,i_1}_{BE_B},M^{2,i_1}_{BE_B}\ldots\}.$$ Probability of outcome $i_2$ is $p_{i_2|i_1}\defeq \Tr(M^{i_2,i_1}_{BE_B}\phi^{i_1}_{BE_B})$. Let $\phi^{i_2,i_1}_{RBACE_AE_B}$ be the global normalized quantum state conditioned on this outcome $i_2$ and previous outcome $i_1$. Bob sends message $i_2$ to Alice. 

\item Consider any odd round $1<k\leq r$. Let the measurement outcomes in previous rounds be $i_1,i_2\ldots i_{k-1}$ and global normalized state be $\phi^{i_{k-1},i_{k-2}\ldots i_1}_{RBACE_AE_B}$. Alice performs the projective measurement $\M^{i_{k-1},i_{k-2}\ldots i_2,i_1}=\{M^{1,i_{k-1},i_{k-2}\ldots i_2,i_1}_{ACE_A},M^{2,i_{k-1},i_{k-2}\ldots i_2,i_1}_{ACE_A}\ldots\}$ and obtains outcome $i_k$ with probability $p_{i_k|i_{k-1},i_{k-2}\ldots i_2,i_1}\defeq\Tr(M^{i_k,i_{k-1},i_{k-2}\ldots i_2,i_1}_{ACE_A}\phi^{i_{k-1},i_{k-2}\ldots i_1}_{AXE_A})$. Let the global normalized state after outcome $i_k$ be $\phi^{i_k,i_{k-1},i_{k-2}\ldots i_1}_{RBACE_BE_A}$. Alice sends the outcome $i_k$ to Bob. 

\item Consider an even round $2<k\leq r$. Let the measurement outcomes in previous rounds be $i_1,i_2\ldots i_{k-1}$ and global normalized state be $\phi^{i_{k-1},i_{k-2}\ldots i_1}_{RBACE_AE_B}$. Bob performs the measurement $$\M^{i_{k-1},i_{k-2}\ldots i_2,i_1}=\{M^{1,i_{k-1},i_{k-2}\ldots i_2,i_1}_{BE_B},M^{2,i_{k-1},i_{k-2}\ldots i_2,i_1}_{BE_B}\ldots\}$$ and obtains outcome $i_k$ with probability $$p_{i_k|i_{k-1},i_{k-2}\ldots i_2,i_1}\defeq\Tr(M^{i_k,i_{k-1},i_{k-2}\ldots i_2,i_1}_{BE_B}\phi^{i_{k-1},i_{k-2}\ldots i_1}_{BE_B}).$$ Let the global normalized state after outcome $i_k$ be $\phi^{i_k,i_{k-1},i_{k-2}\ldots i_1}_{RBACE_BE_A}$. Bob sends the outcome $i_k$ to Alice. 

\item After receiving message $i_r$ from Alice at the end of round $r$, Bob applies a unitary $U^b_{i_r,i_{r-1}\ldots i_1}:BE_B\rightarrow BC_0T_B$ such that $E_B\equiv C_0T_B$ and $C_0\equiv C$. Alice applies a unitary $U^a_{i_r,i_{r-1}\ldots i_1}:ACE_A\rightarrow ACE_A$. Let $U_{i_r,i_{r-1}\ldots i_1}\defeq U^a_{i_r,i_{r-1}\ldots i_1}\otimes U^b_{i_r,i_{r-1}\ldots i_1}$. Define $$\ket{\tau^{i_r,i_{r-1}\ldots i_1}}_{RBACC_0T_BE_A}\defeq U_{i_r,i_{r-1}\ldots i_1}\ket{\phi^{i_r,i_{r-1}\ldots i_1}}_{RBACE_BE_A}.$$

\item For every $k\leq r$, define $$p_{i_1,i_2\ldots i_k}\defeq p_{i_1}\cdot p_{i_2|i_1}\cdot p_{i_3|i_2,i_1}\ldots p_{i_k|i_{k-1},i_{k-2}\ldots i_1}.$$ The joint state in registers $RBC_0A$, after Alice and Bob's final unitaries and averaged over all messages is $\Psi'_{RBC_0A}\defeq\sum_{i_r,i_{r-1}\ldots i_1}p_{i_1,i_2\ldots i_r}\tau^{i_r,i_{r-1}\ldots i_1}_{RBC_0A}$. It satisfies  $\P(\Psi'_{RBC_0A},\Psi_{RBC_0A})\leq \eps$. 
\end{itemize}
\bigskip
\end{mdframed}
\bigskip
The expected communication cost is as follows.
\begin{fact}
\label{expcost}
Expected communication cost of $\mathcal{P}$ is $$\sum_{i_1,i_2\ldots i_r}p_{i_1,i_2\ldots i_r}\log(i_1\cdot i_2\ldots i_r)$$
\end{fact}
\begin{proof}
The expected communication cost is the expected length of the messages over all probability outcomes. It can be evaluated as $$\sum_{i_1}p_{i_1}\log(i_1) + \sum_{i_1,i_2}p_{i_1}p_{i_2|i_1}\log(i_2)+\ldots \sum_{i_1,i_2\ldots i_r}p_{i_1,i_2\ldots i_{r-1}}p_{i_r|i_{r-1},i_{r-2}\ldots i_1}\log(i_r)$$$$ = \sum_{i_1,i_2\ldots i_r}p_{i_1,i_2\ldots i_r}(\log(i_1)+\log(i_1)+\ldots \log(i_r)).$$ 
\end{proof}

This allows us to define

\begin{definition}
\label{def:commweight}
\textbf{Communication weight} of a probability distribution $\{p_1,p_2\ldots p_m\}$  is $\sum_{i=1}^m p_i\log(i)$.
\end{definition}

The following lemma is a coherent representation of above protocol.

\begin{lemma}
\label{cohlemma}
For every $k\leq r$, let $\O_k$ represent the set of all tuples $(i_1,i_2\ldots i_k)$ which satisfy: $\{i_1,i_2\ldots i_k\}$ is a sequence of measurement outcomes that occurs with non-zero probability upto $k$-th round of $\mathcal{P}$. 

There exist registers $M_1,M_2\ldots M_r$ and isometries $$\{U_{i_{k-1},i_{k-2}\ldots i_2,i_1}: ACE_A\rightarrow ACE_AM_k| k >1, k \text{ odd }, (i_1,i_2\ldots i_{k-1})\in \O_{k-1}\},$$ $$\{U_{i_{k-1},i_{k-2}\ldots i_2,i_1}: BE_B\rightarrow BE_BM_k|  k \text{ even }, (i_1,i_2\ldots i_{k-1})\in \O_{k-1}\}$$ and $U: ACE_A\rightarrow ACE_AM_1$, such that 
$$\ket{\Psi}_{RBCA}\ket{\theta}_{E_AE_B} = U^{\dagger}\sum_{i_1,i_2\ldots i_r}\sqrt{p_{i_1,i_2\ldots i_r}}U^{\dagger}_{ i_1}U^{\dagger}_{ i_2,i_1}\ldots U^{\dagger}_{i_r,i_{r-1}\ldots i_1}\ket{\tau^{i_r,i_{r-1}\ldots i_1}}_{RBCAC_0T_BE_A}\ket{i_r}_{M_r}\ldots\ket{i_1}_{M_1}.$$ 
\end{lemma} 

\begin{proof}

Fix an odd $k>1$. Let the messages prior to $k-$th round be $(i_1,i_2\ldots i_{k-1})$. As defined in protocol $\mathcal{P}$, global quantum state before $k$-th round is $\phi^{i_{k-1},i_{k-2}\ldots i_1}_{RBCAE_AE_B}$. Alice performs the measurement $$\{M^{1,i_{k-1},i_{k-2}\ldots i_2,i_1}_{ACE_A},M^{2,i_{k-1},i_{k-2}\ldots i_2,i_1}_{AXE_A}\ldots\}.$$ This leads to a \textit{convex-split} (introduced in \cite{jain14}):
\begin{eqnarray}
\label{roundconvsplit}
\phi^{i_{k-1},i_{k-2}\ldots i_1}_{RBE_B} &=& \sum_{i_k} \Tr_{ACE_A}(M^{i_k,i_{k-1},i_{k-2}\ldots i_2,i_1}_{ACE_A}\phi^{i_{k-1},i_{k-2}\ldots i_1}_{RBCAE_BE_A}) \nonumber\\&=& \sum_{i_k} p_{i_k|i_{k-1},i_{k-2}\ldots i_2,i_1}\frac{\Tr_{ACE_A}(M^{i_k,i_{k-1},i_{k-2}\ldots i_2,i_1}_{ACE_A}\phi^{i_{k-1},i_{k-2}\ldots i_1}_{RBCAE_BE_A}M^{i_k,i_{k-1},i_{k-2}\ldots i_2,i_1}_{ACE_A})}{p_{i_k|i_{k-1},i_{k-2}\ldots i_2,i_1}}\nonumber\\&=& \sum_{i_k} p_{i_k|i_{k-1},i_{k-2}\ldots i_2,i_1}\phi^{i_k,i_{k-1},i_{k-2}\ldots i_2,i_1}_{RBE_B}
\end{eqnarray}

A purification of $\phi^{i_{k-1},i_{k-2}\ldots i_1}_{RBE_B}$ on registers $RBCAE_BE_A$ is $\phi^{i_{k-1},i_{k-2}\ldots i_1}_{RBCAE_BE_A}$. Introduce a register $M_{k}$ (of sufficiently large dimension) and consider the following purification of $$\sum_{i_k} p_{i_k|i_{k-1},i_{k-2}\ldots i_2,i_1}\phi^{i_k,i_{k-1},i_{k-2}\ldots i_2,i_1}_{RBE_B}$$ on register $RBCAE_BE_AM_k$ : $$\sum_{i_k}\sqrt{p_{i_k|i_{k-1},i_{k-2}\ldots i_2,i_1}}\ket{\phi^{i_k,i_{k-1},i_{k-2}\ldots i_2,i_1}}_{RBCAE_BE_A}\ket{i_k}_{M_k}.$$

By Uhlmann's theorem \ref{uhlmann}, there exists an isometry $U_{i_{k-1},i_{k-2}\ldots i_2,i_1}: ACE_A\rightarrow ACE_AM_k$ such that
\begin{equation}
\label{aliceunitary} 
U_{i_{k-1},i_{k-2}\ldots i_2,i_1}\ket{\phi^{i_{k-1},i_{k-2}\ldots i_1}}_{RBCAE_BE_A} = \sum_{i_k}\sqrt{p_{i_k|i_{k-1},i_{k-2}\ldots i_2,i_1}}\ket{\phi^{i_k,i_{k-1},i_{k-2}\ldots i_2,i_1}}_{RBCAE_BE_A}\ket{i_k}_{M_k}
\end{equation}

For $k=1$, introduce register $M_1$ of sufficiently large dimension. Similar argument implies that there exists an isometry $U: ACE_A\rightarrow ACE_AM_1$ such that

\begin{equation}
\label{aliceunitary1} 
U\ket{\Psi}_{RBACE_BE_A} = \sum_{i_1}\sqrt{p_{i_1}}\ket{\phi^{i_1}}_{RBACE_BE_A}\ket{i_1}_{M_1}
\end{equation}

For $k$ even, introduce a register $M_k$ of sufficiently large dimension. Again by similar argument, there exists an isometry $U_{i_{k-1},i_{k-2}\ldots i_2,i_1}: BE_B\rightarrow BE_BM_k$ such that 

\begin{equation}
\label{bobunitary} 
U_{i_{k-1},i_{k-2}\ldots i_2,i_1}\ket{\phi^{i_{k-1},i_{k-2}\ldots i_1}}_{RBCAE_BE_A} = \sum_{i_k}\sqrt{p_{i_k|i_{k-1},i_{k-2}\ldots i_2,i_1}}\ket{\phi^{i_k,i_{k-1},i_{k-2}\ldots i_2,i_1}}_{RBCAE_BE_A}\ket{i_k}_{M_k}
\end{equation}

Now, we recursively use equations \ref{aliceunitary}, \ref{aliceunitary1} and \ref{bobunitary}. Consider,
\begin{eqnarray*}
&&\ket{\Psi}_{RBCA}\ket{\theta}_{E_AE_B} = U^{\dagger}\sum_{i_1}\sqrt{p_{i_1}}\ket{\phi^{i_1}}_{RBCAE_BE_A}\ket{i_1}_{M_1} \\ &=& U^{\dagger}\sum_{i_1}\sqrt{p_{i_1}}U^{\dagger}_{i_1}\sum_{i_2}\sqrt{p_{i_2|i_1}}\ket{\phi^{i_2,i_1}}_{RBCAE_BE_A}\ket{i_2}_{M_2}\ket{i_1}_{M_1}\\ &=& U^{\dagger}\sum_{i_1,i_2}\sqrt{p_{i_1,i_2}}U_{i_1}^{\dagger}\ket{\phi^{i_2,i_1}}_{RBCAE_BE_A}\ket{i_2}_{M_2}\ket{i_1}_{M_1} \\&=& U^{\dagger}\sum_{i_1,i_2\ldots i_r}\sqrt{p_{i_1,i_2\ldots i_r}}U^{\dagger}_{i_1}U^{\dagger}_{i_2,i_1}\ldots U^{\dagger}_{i_r,i_{r-1}\ldots i_1}\ket{\tau^{i_r,i_{r-1}\ldots i_1}}_{RBCAB_0T_BE_A}\ket{i_r}_{M_r}\ldots\ket{i_1}_{M_1}
\end{eqnarray*}

Last equality follows by recursion. This completes the proof.

\end{proof}

\begin{definition}
\label{shortunitaries}
We introduce the following useful definitions.

\begin{itemize}
\item Let $k>1$ be odd. Isometry $U_k: ACE_AM_1M_2\ldots M_{k-1}\rightarrow ACE_AM_1M_2\ldots M_{k-1}M_k$,  $$U_k \defeq \sum_{i_1,i_2\ldots i_{k-1}} \ketbra{i_1}_{M_1}\otimes \ketbra{i_2}_{M_2}\otimes\ldots\ketbra{i_{k-1}}_{M_{k-1}}\otimes U_{i_{k-1},i_{k-2}\ldots i_2,i_1}.$$
\item  For $k$ even, Isometry $U_k: BE_BM_1M_2\ldots M_{k-1}\rightarrow BE_BM_1M_2\ldots M_{k-1}M_k$, $$U_k \defeq \sum_{i_1,i_2\ldots i_{k-1}} \ketbra{i_1}_{M_1}\otimes \ketbra{i_2}_{M_2}\otimes\ldots\ketbra{i_{k-1}}_{M_{k-1}}\otimes U_{i_{k-1},i_{k-2}\ldots i_2,i_1}.$$
\item Unitary $U^a_{r+1}: ACE_AM_1M_2\ldots M_r \rightarrow ACE_AM_1M_2\ldots M_r$, $$U^a_{r+1} \defeq \sum_{i_1,i_2\ldots i_r} \ketbra{i_1}_{M_1}\otimes \ketbra{i_2}_{M_2}\otimes\ldots\ketbra{i_r}_{M_r}\otimes U^a_{i_r,i_{r-1}\ldots i_1}.$$
\item Unitary $U^b_{r+1}: BE_BM_1M_2\ldots M_r \rightarrow BC_0T_BM_1M_2\ldots M_r$, $$U^b_{r+1} \defeq \sum_{i_1,i_2\ldots i_r} \ketbra{i_1}_{M_1}\otimes \ketbra{i_2}_{M_2}\otimes\ldots\ketbra{i_r}_{M_r}\otimes U^b_{i_r,i_{r-1}\ldots i_1}.$$
\item Unitary $U_{r+1}:ACE_ABE_BM_1M_2\ldots M_r \rightarrow ACE_ABC_0T_BM_1M_2\ldots M_r$, $$U_{r+1} \defeq \sum_{i_1,i_2\ldots i_r} \ketbra{i_1}_{M_1}\otimes \ketbra{i_2}_{M_2}\otimes\ldots\ketbra{i_r}_{M_r}\otimes U_{i_r,i_{r-1}\ldots i_1}.$$
\end{itemize}

\end{definition}

This leads to a more convenient representation of lemma \ref{cohlemma}.  
\begin{cor}
\label{cohequation}
It holds that 
$$\ket{\Psi}_{RBCA}\ket{\theta}_{E_AE_B}=U^{\dagger}U_2^{\dagger}\ldots U_{r+1}^{\dagger} \sum_{i_1,i_2\ldots i_r}\sqrt{p_{i_1,i_2\ldots i_r}}\ket{\tau^{i_r,i_{r-1}\ldots i_1}}_{RBCAC_0T_BE_A}\ket{i_r}_{M_r}\ldots\ket{i_1}_{M_1}.$$ and 
$$\P(\Psi_{RBC_0A},\sum_{i_1,i_2\ldots i_r}p_{i_1,i_2\ldots i_r}\tau^{i_r,i_{r-1}\ldots i_1}_{RBC_0A})\leq \eps.$$
\end{cor}

\begin{proof}

The corollary follows immediately using Definition \ref{shortunitaries} and lemma \ref{cohlemma}.

\end{proof}

\section{Lower bound on expected communication cost}
\label{sec:lowerbound}

In this section, we obtain a lower bound on expected communication cost of quantum state redistribution and quantum state transfer, by considering a class of states defined below.

Let register $R$ be composed of two registers $R_A,R'$, such that $R\equiv R_AR'$. Let $d_a$ be the dimension of registers $R_A$ and $A$. Let $d$  be the dimension of registers $R',C$ and $B$. Consider,
\begin{definition}
\label{staterediststate}
 $\ket{\Psi}_{RBCA}\defeq\frac{1}{\sqrt{d_a}}\sum_{a=1}^{d_a}\ket{a}_{R_A}\ket{a}_A\ket{\psi^a}_{R'BC}$. Let $\ket{\psi^a}_{R'BC}=\sum_{j=1}^d\sqrt{e_j}\ket{u_j}_{R'}\ket{v_j(a)}_B\ket{w_j(a)}_C$ where $e_1\geq e_2\geq \ldots e_d>0$, $\sum_{i=1}^d e_i = 1$ and $\{\ket{u_1},\ldots\ket{u_d}\}$, $\{\ket{v_1(a)},\ldots\ket{v_d(a)}\}$, $\{\ket{w_1(a)},\ldots\ket{w_d(a)}\}$ form an orthonormal basis (second and third bases may depend arbitrarily on $a$) in their respective Hilbert spaces. 
\end{definition}

For quantum state transfer, we consider a pure state $\tilde{\Psi}_{RC}$ with Schmidt decomposition $\sum_{j=1}^d \sqrt{e_j}\ket{u_j}_R\ket{w_j}_C$. 
 
Given the state $\psi^a_{R'BC}$ from definition \ref{staterediststate}, we define a `GHZ state' corresponding to it: $\ket{\omega^a}_{R'BC}\defeq \frac{1}{\sqrt{d}}\sum_{j=1}^d\ket{u_j}_{R'}\ket{v_j(a)}_B\ket{w_j(a)}_C$. Using this, we define $\omega_{RBCA}\defeq \frac{1}{\sqrt{d_a}}\sum_{a=1}^{d_a}\ket{a}_{R_A}\ket{a}_A\ket{\omega^a}_{R'BC}$. Similarly, given the bipartite state $\tilde{\Psi}_{RC}$, we define a maximally entangled state $\omega'_{RC}\defeq \frac{1}{\sqrt{d}}\sum_{j=1}^d\ket{u_j}_R\ket{w_j}_C$.

Following two relations are easy to verify. 
\begin{equation}
\label{psiandomega}
\ket{\omega}_{RBCA} = \frac{1}{\sqrt{d_a\cdot d}}\Psi_R^{-\frac{1}{2}}\ket{\Psi}_{RBCA} \text{ and } \ket{\omega'}_{RC} = \frac{1}{\sqrt{d}}(\tilde{\Psi}_R)^{-\frac{1}{2}}\ket{\tilde{\Psi}}_{RC}
\end{equation}

\bigskip

As noted in section \ref{sec:cohtrans}, the protocol $\mathcal{P}$ achieves quantum state redistribution of $\Psi_{RBCA}$ with error $\eps$ and expected communication cost $C$. Following lemma is a refined form of corollary \ref{cohequation}, and is also applicable to state $\Psi_{RBCA}$ not of the form given in definition \ref{staterediststate}.  

\begin{lemma}
\label{goodcoh}
There exists a probability distribution $\{p'_{i_1,i_2\ldots i_r}\}$ and pure states $\kappa^{i_r,i_{r-1}\ldots i_1}_{CE_AT_B}$ such that 

$$\P(\Psi_{RBCA}\otimes\theta_{E_AE_B},U^{\dagger}U_2^{\dagger}\ldots U_{r+1}^{\dagger}\sum_{i_1,i_2\ldots i_r}\sqrt{p'_{i_1,i_2\ldots i_r}}\Psi_{RBC_0A}\otimes\kappa^{i_r,i_{r-1}\ldots i_1}_{CE_AT_B}\ket{i_r}_{M_r}\ldots\ket{i_1}_{M_1})\leq 2\sqrt{\eps},$$ and the communication weight of  $p'_{i_1,i_2\ldots i_r}$ is at most $\frac{C}{1-\eps}$. 
\end{lemma}

\begin{proof}
Let $\B$ be the set of tuples $(i_1,i_2\ldots i_r)$ for which $\F^2(\Psi_{RBC_0A},\tau^{i_r,i_{r-1}\ldots i_1}_{RBC_0A})\leq 1-\eps$. Let $\G$ be remaining set of tuples. From corollary \ref{cohequation} and purity of $\Psi_{RBC_0A}$, it holds that
$$\sum_{i_1,i_2\ldots i_r}p_{i_1,i_2\ldots i_r}\F^2(\Psi_{RBC_0A},\tau^{i_r,i_{r-1}\ldots i_1}_{RBC_0A})\geq 1-\eps^2.$$

Thus, $$(1-\eps)\sum_{(i_1,i_2\ldots i_r)\in\B}p_{i_1,i_2\ldots i_r}+ \sum_{(i_1,i_2\ldots i_r)\in\G}p_{i_1,i_2\ldots i_r} \geq 1-\eps^2,$$ which implies $\sum_{(i_1,i_2\ldots i_r)\in\B}p_{i_1,i_2\ldots i_r}\leq \eps$. Thus we have $\sum_{(i_1,i_2\ldots i_r)\in \G}p_{i_1,i_2\ldots i_r}\geq 1-\eps$. 

Define $p'_{i_1,i_2\ldots i_r}\defeq \frac{p_{i_1,i_2\ldots i_r}}{\sum_{{i_1,i_2\ldots i_r}\in \G} p_{i_1,i_2\ldots i_r}}$, if $(i_1,i_2\ldots i_r) \in \G$ and $p'_{i_1,i_2\ldots i_r}\defeq 0$ if $(i_1,i_2\ldots i_r)\in \B $. 
\bigskip

For all $(i_1,i_2\ldots i_r)\in \G$, $\F^2(\Psi_{RBC_0A},\tau^{i_r,i_{r-1}\ldots i_1}_{RBC_0A})\geq 1-\eps$. Thus by Fact \ref{uhlmann},  there exists a pure state $\kappa^{i_r,i_{r-1}\ldots i_1}_{CE_AT_B}$ such that 

\begin{equation}
\label{goodproperty} 
\F^2(\Psi_{RBC_0A}\otimes\kappa^{i_r,i_{r-1}\ldots i_1}_{CE_AT_B},\tau^{i_r,i_{r-1}\ldots i_1}_{RBCAC_0T_BE_A})\geq 1-\eps 
\end{equation}

Consider,

\begin{eqnarray*}
&&\P(\sum_{i_1,i_2\ldots i_r}\sqrt{p_{i_1,i_2\ldots i_r}}\tau^{i_r,i_{r-1}\ldots i_1}_{RBCAC_0T_BE_A}\ket{i_r}_{M_r}\ldots\ket{i_1}_{M_1},\sum_{i_1,i_2\ldots i_r}\sqrt{p'_{i_1,i_2\ldots i_r}}\tau^{i_r,i_{r-1}\ldots i_1}_{RBCAC_0T_BE_A}\ket{i_r}_{M_r}\ldots\ket{i_1}_{M_1})\nonumber \\&=& \sqrt{1-(\sum_{i_1,i_2\ldots i_r}\sqrt{p_{i_1,i_2\ldots i_r}p'_{i_1,i_2\ldots i_r}})^2} = \sqrt{1-(\sum_{i_1,i_2\ldots i_r\in \G}p_{i_1,i_2\ldots i_r})}\leq \sqrt{\eps} 
\end{eqnarray*}

 and 
\begin{eqnarray*}
&&\P(\sum_{i_1,i_2\ldots i_r}\sqrt{p'_{i_1,i_2\ldots i_r}}\tau^{i_r,i_{r-1}\ldots i_1}_{RBCAC_0T_BE_A}\ket{i_r}_{M_r}\ldots\ket{i_1}_{M_1}, \sum_{i_1,i_2\ldots i_r}\sqrt{p'_{i_1,i_2\ldots i_r}}\Psi_{RBC_0A}\otimes\kappa^{i_r,i_{r-1}\ldots i_1}_{CE_AT_B}\ket{i_r}_{M_r}\ldots\ket{i_1}_{M_1}) \\&=& \sqrt{1-(\sum_{i_1,i_2\ldots i_r}p'_{i_1,i_2\ldots i_r}\F(\tau^{i_r,i_{r-1}\ldots i_1}_{RBCAC_0T_BE_A},\Psi_{RBC_0A}\otimes\kappa^{i_r,i_{r-1}\ldots i_1}_{CE_AT_B}))^2} \leq \sqrt{\eps} \quad (\text{Equation \ref{goodproperty}})
\end{eqnarray*}

These together imply, using triangle inequality for purified distance (Fact \ref{fact:trianglepurified}),

\begin{eqnarray*}
&&\P(\sum_{i_1,i_2\ldots i_r}\sqrt{p_{i_1,i_2\ldots i_r}}\tau^{i_r,i_{r-1}\ldots i_1}_{RBCAC_0T_BE_A}\ket{i_r}_{M_r}\ldots\ket{i_1}_{M_1},\sum_{i_1,i_2\ldots i_r}\sqrt{p'_{i_1,i_2\ldots i_r}}\Psi_{RBC_0A}\otimes\kappa^{i_r,i_{r-1}\ldots i_1}_{CE_AT_B}\ket{i_r}_{M_r}\ldots\ket{i_1}_{M_1})\\&\leq& 2\sqrt{\eps} 
\end{eqnarray*}

Thus, from corollary \ref{cohequation}, we have $$\P(\Psi_{RBCA}\otimes\theta_{E_AE_B},U^{\dagger}U_2^{\dagger}\ldots U_{r+1}^{\dagger}\sum_{i_1,i_2\ldots i_r}\sqrt{p'_{i_1,i_2\ldots i_r}}\Psi_{RBC_0A}\otimes\kappa^{i_r,i_{r-1}\ldots i_1}_{CE_AT_B}\ket{i_r}_{M_r}\ldots\ket{i_1}_{M_1})\leq 2\sqrt{\eps}.$$

The communication weight of $p'_{i_1,i_2\ldots i_r}$ is 
\begin{eqnarray*}
\sum_{i_1,i_2\ldots i_r} p'_{i_1,i_2\ldots i_r}\log(i_1\cdot i_2\ldots i_r) &\leq& \frac{1}{1-\eps}\sum_{{i_1,i_2\ldots i_r}\in \G}p_{i_1,i_2\ldots i_r}\log(i_1\cdot i_2\ldots i_r) \\ &\leq& \frac{1}{1-\eps}\sum_{i_1,i_2\ldots i_r}p_{i_1,i_2\ldots i_r}\log(i_1\cdot i_2\ldots i_r) =\frac{C}{1-\eps}.
 \end{eqnarray*}
This completes the proof.
\end{proof}

We now use Lemma \ref{goodcoh} to prove the following for the state $\omega_{RBCA}$. Recall that $e_d$ is the smallest eigenvalue of $\psi^a_{R'}$, independent of $a$.

\begin{lemma}
\label{convepr}
It holds that $$\P(\omega_{RBCA}\otimes\theta_{E_AE_B},U^{\dagger}U_2^{\dagger}\ldots U_{r+1}^{\dagger}\sum_{i_1,i_2\ldots i_r}\sqrt{p'_{i_1,i_2\ldots i_r}}\omega_{RBC_0A}\otimes\kappa^{i_r,i_{r-1}\ldots i_1}_{CE_AT_B}\ket{i_r}_{M_r}\ldots\ket{i_1}_{M_1})\leq \sqrt{\frac{8\eps}{e_d\cdot d}}.$$

Communication weight of distribution $p'_{i_1,i_2\ldots i_r}$ is $\frac{C}{1-\eps}$.  
\end{lemma}
\begin{proof}

Define a completely positive map $\tilde{\E}:R\rightarrow R$ as $ \tilde{\E}(\rho)\defeq \frac{e_d}{d_a}(\Psi^{-\frac{1}{2}}_R\rho\Psi^{-\frac{1}{2}}_R)$, which is trace non-increasing since $\Psi^{-1}_R \leq \frac{d_a}{e_d}\text{I}_R$. Using equation \ref{psiandomega}, observe that $$\tilde{\E}(\Psi_{RBCA}) = e_d\cdot d\cdot\omega_{RBCA}.$$

Consider,
\begin{eqnarray*}
2\sqrt{\eps} &\geq& \P(\Psi_{RBCA}\otimes\theta_{E_AE_B},U^{\dagger}U_2^{\dagger}\ldots U_{r+1}^{\dagger}\sum_{i_1,i_2\ldots i_r}\sqrt{p'_{i_1,i_2\ldots i_r}}\Psi_{RBC_0A}\otimes\kappa^{i_r,i_{r-1}\ldots i_1}_{CE_AT_B}\ket{i_r}_{M_r}\ldots\ket{i_1}_{M_1})\\ &&\text{(Lemma \ref{goodcoh})}\\&\geq& \P(\tilde{\E}(\Psi_{RBCA})\otimes\theta_{E_AE_B},U^{\dagger}U_2^{\dagger}\ldots U_{r+1}^{\dagger}\sum_{i_1,i_2\ldots i_r}\sqrt{p'_{i_1,i_2\ldots i_r}}\tilde{\E}(\Psi_{RBC_0A})\otimes\kappa^{i_r,i_{r-1}\ldots i_1}_{CE_AT_B}\ket{i_r}_{M_r}\ldots\ket{i_1}_{M_1})\\ && (\text{Fact \ref{fact:monotonequantumoperation}}) \\ &=& \P(d\cdot e_d\cdot\omega_{RBCA}\otimes\theta_{E_AE_B},d\cdot e_d\cdot U^{\dagger}U_2^{\dagger}\ldots U_{r+1}^{\dagger}\sum_{i_1,i_2\ldots i_r}\sqrt{p'_{i_1,i_2\ldots i_r}}\omega_{RBC_0A}\otimes\kappa^{i_r,i_{r-1}\ldots i_1}_{CE_AT_B}\ket{i_r}_{M_r}\ldots\ket{i_1}_{M_1})
\end{eqnarray*}

Using Fact \ref{scalarpurified}, we thus obtain $$\P(\omega_{RBCA}\otimes\theta_{E_AE_B}, U^{\dagger}U_2^{\dagger}\ldots U_{r+1}^{\dagger}\sum_{i_1,i_2\ldots i_r}\sqrt{p'_{i_1,i_2\ldots i_r}}\omega_{RBC_0A}\otimes\kappa^{i_r,i_{r-1}\ldots i_1}_{CE_AT_B}\ket{i_r}_{M_r}\ldots\ket{i_1}_{M_1})\leq \sqrt{\frac{8\eps}{d\cdot e_d}}.$$

Furthermore, there is no change in communication weight. This completes the proof. 

\end{proof}

Similarly for quantum state transfer, we have the following corollary 
\begin{cor}
\label{conveprmerge}
It holds that $$\P(\omega'_{RC}\otimes\theta_{E_AE_B},U^{\dagger}U_2^{\dagger}\ldots U_{r+1}^{\dagger}\sum_{i_1,i_2\ldots i_r}\sqrt{p'_{i_1,i_2\ldots i_r}}\omega'_{RC_0}\otimes\kappa^{i_r,i_{r-1}\ldots i_1}_{CE_AT_B}\ket{i_r}_{M_r}\ldots\ket{i_1}_{M_1})\leq \sqrt{\frac{8\eps}{e_d\cdot d}}.$$

Communication weight of distribution $p'_{i_1,i_2\ldots i_r}$ is $\frac{C}{1-\eps}$.  
\end{cor}

Now we exhibit an interactive entanglement assisted communication protocol for state-redistribution of $\omega_{RBCA}$ with suitably upper bounded worst case communication cost.

\begin{lemma}
\label{exptoworst}
Fix an error parameter $\mu>0$. There exists an entanglement assisted $r$-round quantum communication protocol for state redistribution of $\omega_{RBCA}$ with worst case quantum communication cost at most $\frac{2C}{\mu(1-\eps)}$ and error at most $\sqrt{\frac{8\eps}{e_d\cdot d}}+\sqrt{\mu}$.
\end{lemma}

\begin{proof}
From lemma \ref{convepr}, we have that $$\P(\omega_{RBCA}\otimes\theta_{E_AE_B},U^{\dagger}U_2^{\dagger}\ldots U_{r+1}^{\dagger}\sum_{i_1,i_2\ldots i_r}\sqrt{p'_{i_1,i_2\ldots i_r}}\omega_{RBC_0A}\otimes\kappa^{i_r,i_{r-1}\ldots i_1}_{CE_AT_B}\ket{i_r}_{M_r}\ldots\ket{i_1}_{M_1})\leq \sqrt{\frac{8\eps}{a_d\cdot d}},$$ and $$\sum_{i_1,i_2\ldots i_r}p'_{i_1,i_2\ldots i_r}\log(i_1\cdot i_2\ldots i_r)\leq \frac{C}{1-\eps}.$$

Consider the set of tuples $(i_1,i_2\ldots i_r)$ which satisfy $i_1\cdot i_2\ldots i_r>2^{\frac{C}{(1-\eps)\mu}}$. Let this set be $\B'$ and $\G'$ be the set of rest of the tuples. Then $$\frac{C}{(1-\eps)} > \sum_{i_1,i_2\ldots i_r\in \B'}p'_{i_1,i_2\ldots i_r}\log(i_1\cdot i_2\ldots i_r) > \frac{C}{(1-\eps)\mu}\sum_{i_1,i_2\ldots i_r\in \B'}p'_{i_1,i_2\ldots i_r}.$$ This implies $\sum_{i_1,i_2\ldots i_r\in \B'}p'_{i_1,i_2\ldots i_r} < \mu$.
 Define a new probability distribution $q_{i_1,i_2\ldots i_r}\defeq \frac{p'_{i_1,i_2\ldots i_r}}{\sum_{(i_1,i_2\ldots i_r)\in \G'}p'_{i_1,i_2\ldots i_r}}$ for all $(i_1,i_2\ldots i_r)\in \G'$ and $q_{i_1,i_2\ldots i_r}=0$ for all $(i_1,i_2\ldots i_r)\in \B'$. Consider,

$$\P(\sum_{i_1,i_2\ldots i_r}\sqrt{p'_{i_1,i_2\ldots i_r}}\omega_{RBC_0A}\otimes\kappa^{i_r,i_{r-1}\ldots i_1}_{CE_AT_B}\ket{i_r}_{M_r}\ldots\ket{i_1}_{M_1},\sum_{i_1,i_2\ldots i_r}\sqrt{q_{i_1,i_2\ldots i_r}}\omega_{RBC_0A}\otimes\kappa^{i_r,i_{r-1}\ldots i_1}_{CE_AT_B}\ket{i_r}_{M_r}\ldots\ket{i_1}_{M_1})$$ $$= \sqrt{1-(\sum_{i_1,i_2\ldots i_r}\sqrt{p'_{i_1,i_2\ldots i_r}q_{i_1,i_2\ldots i_r}})^2} = \sqrt{1-\sum_{(i_1,i_2\ldots i_r) \in \G'}p'_{i_1,i_2\ldots i_r}}\leq \sqrt{\mu}.$$Thus, triangle inequality for purified distance (Fact \ref{fact:trianglepurified}) implies

\begin{eqnarray*}
&&\P(\omega_{RBCA}\otimes\theta_{E_AE_B},U^{\dagger}U_2^{\dagger}\ldots U_{r+1}^{\dagger}\sum_{i_1,i_2\ldots i_r}\sqrt{q_{i_1,i_2\ldots i_r}}\omega_{RBC_0A}\otimes\kappa^{i_r,i_{r-1}\ldots i_1}_{CE_AT_B}\ket{i_r}_{M_r}\ldots\ket{i_1}_{M_1})\\&\leq& \sqrt{\frac{8\eps}{e_d\cdot d}}+\sqrt{\mu}
\end{eqnarray*}
Defining $\pi_{RBCAE_AE_B}\defeq U^{\dagger}U_2^{\dagger}\ldots U_{r+1}^{\dagger}\sum_{i_1,i_2\ldots i_r\in \G'}\sqrt{q_{i_1,i_2\ldots i_r}}\omega_{RBC_0A}\otimes\kappa^{i_r,i_{r-1}\ldots i_1}_{CE_AT_B}\ket{i_r}_{M_r}\ldots\ket{i_1}_{M_1}$, we have
\begin{equation}
\label{eq:closegoodstate}
\P(\omega_{RBCA}\otimes\theta_{E_AE_B},\omega'_{RBCE_AE_B})\leq \sqrt{\frac{8\eps}{e_d\cdot d}}+\sqrt{\mu}
\end{equation} 

\bigskip

Let $\T$ be the set of all tuples $(i_1,i_2\ldots i_k)$ (with $k\leq r$) that satisfy the following property: there exists a set of positive integers $\{i_{k+1},i_{k+2}\ldots i_r\}$ such that $(i_1,i_2\ldots i_k,i_{k+1}\ldots i_r)\in \G'$. Consider the following protocol $\mathcal{P'}$.
\bigskip
\begin{mdframed}
\bigskip

\textbf{Input:} A quantum state in registers $RBCAE_AE_B$.

\begin{itemize}
\item Alice applies the isometry $U:ACE_A\rightarrow ACE_AM_1$ (definition \ref{shortunitaries}). She introduces a register $M'_1\equiv M_1$ in the state $\ket{0}_{M'_1}$ and performs the following unitary $W_1: M_1M'_1\rightarrow M_1M'_1$: $$W_1\ket{i}_{M_1}\ket{0}_{M'_1}=\ket{i}_{M_1}\ket{i}_{M'_1} \quad \text{if } (i)\in \T \quad \text{and}\quad W_1\ket{i}_{M_1}\ket{0}_{M'_1}=\ket{i}_{M_1}\ket{0}_{M'_1} \quad \text{if } (i)\notin \T.$$ 
She sends $M'_1$ to Bob.
\item Bob introduces a register $M'_2\equiv M_2$ in the state $\ket{0}_{M'_2}$. If he receives $\ket{0}_{M'_1}$ from Alice, he performs no operation. Else he applies the isometry $U_2: BE_BM'_1\rightarrow BE_BM'_1M_2$ and then performs the following unitary $W_2: M'_1M_2M'_2\rightarrow M'_1M_2M'_2$: $$W_1\ket{i}_{M'_1}\ket{j}_{M_2}\ket{0}_{M'_2}=\ket{i}_{M'_1}\ket{j}_{M_2}\ket{j}_{M'_2} \quad \text{if } (i,j)\in \T$$ and  $$W_1\ket{i}_{M'_1}\ket{j}_{M_2}\ket{0}_{M'_2}=\ket{i}_{M'_1}\ket{j}_{M_2}\ket{0}_{M'_2} \quad \text{if }(i,j)\notin \T.$$ 

He sends $M'_2$ to Alice. 

\item For every odd round $k>1$, Alice introduces a register $M'_k\equiv M_k$ in the state $\ket{0}_{M'_k}$. If she receives $\ket{0}_{M'_{k-1}}$ from Bob, she performs no further operation. Else, she applies the isometry $$U_k: ACE_AM_1M'_2M_3\ldots M'_{k-1}\rightarrow ACE_AM_1M'_2M_3\ldots M'_{k-1}M_k$$ and performs the following unitary $W_k: M_1M'_2\ldots M'_{k-1}M_kM'_k\rightarrow M_1M'_2\ldots M'_{k-1}M_kM'_k$: 
$$W_k\ket{i_1}_{M_1}\ket{i_2}_{M'_2}\ldots\ket{i_k}_{M_k}\ket{0}_{M'_k}=
\ket{i_1}_{M_1}\ket{i_2}_{M'_2}\ldots\ket{i_k}_{M_k}\ket{i_k}_{M'_k} \quad \text{if } (i_1,i_2\ldots i_k)\in \T$$ and  $$W_k\ket{i_1}_{M_1}\ket{i_2}_{M'_2}\ldots\ket{i_k}_{M_k}\ket{0}_{M'_k}=
\ket{i_1}_{M_1}\ket{i_2}_{M'_2}\ldots\ket{i_k}_{M_k}\ket{0}_{M'_k} \quad \text{if }(i_1,i_2\ldots i_k)\notin \T.$$ 

She sends $M'_k$ to Bob.

\item For every even round $k>2$, Bob introduces a register $M'_k\equiv M_k$ in the state $\ket{0}_{M'_k}$. If he receives $\ket{0}_{M'_{k-1}}$ from Alice, he performs no further operation.. Else, he applies the isometry $U_k: BE_BM'_1M_2M'_3\ldots M'_{k-1}\rightarrow BE_BM'_1M_2M'_3\ldots M'_{k-1}M_k$ and performs the following unitary $W_k: M'_1M_2\ldots M'_{k-1}M_kM'_k\rightarrow M'_1M_2\ldots M'_{k-1}M_kM'_k$: 
$$W_k\ket{i_1}_{M'_1}\ket{i_2}_{M_2}\ldots\ket{i_k}_{M_k}\ket{0}_{M'_k}=
\ket{i_1}_{M'_1}\ket{i_2}_{M_2}\ldots\ket{i_k}_{M_k}\ket{i_k}_{M'_k} \quad \text{if } (i_1,i_2\ldots i_k)\in \T$$ and  $$W_k\ket{i_1}_{M'_1}\ket{i_2}_{M_2}\ldots\ket{i_k}_{M_k}\ket{0}_{M'_k}=
\ket{i_1}_{M'_1}\ket{i_2}_{M_2}\ldots\ket{i_k}_{M_k}\ket{0}_{M'_k} \quad \text{if }(i_1,i_2\ldots i_k)\notin \T.$$ 

He sends $M'_k$ to Alice.

\item After round $r$, if Bob receives $\ket{0}_{M'_r}$ from Alice, he performs no further operation. Else he applies the unitary $U^b_{r+1}: BE_BM'_1M_2M'_3\ldots M'_r\rightarrow BC_0T_BM'_1M_2M'_3\ldots M'_r$. Alice applies the unitary $U^a_{r+1}: ACE_AM_1M'_2M_3\ldots M_r\rightarrow ACE_AM_1M'_2M_3\ldots M_r$. They trace out all of their registers except $A,B,C_0$. 
\end{itemize}
\bigskip
\end{mdframed}
\bigskip
Let $\E: RBCAE_AE_B\rightarrow RBC_0A$ be the quantum map generated by $\mathcal{P'}$. For any $k$, if any of the parties receive the state $\ket{0}_{M'_k}$, let this event be called \textit{abort}.

We show the following claim.
\begin{claim} It holds that
 $\E(\pi_{RBCAE_AE_B})=\omega_{RBC_0A}$
\end{claim}
\begin{proof}
We argue that the protocol never aborts when acting on $\pi_{RBCAE_AE_B}$. Consider the first round of the protocol. Define the projector $\Pi\defeq \sum_{i: (i)\notin \T}\ketbra{i}_{M_1}$. From definition \ref{shortunitaries}, it is clear that the isometry $U^{\dagger}_2U^{\dagger}_3\ldots U^{\dagger}_{r+1}$ is of the form $\sum_i \ketbra{i}_{M_1}\otimes V_i$, for some set of isometries $\{V_i\}$ . Thus, from the definition of $\pi_{RBCAE_AE_B}$ (in which the summation is only over the tuples $(i_1,i_2\ldots i_r)\in \G'$), it holds that $$\Pi U\pi_{RBCAE_AE_B}=0.$$  This implies that Bob does not receive the state $\ket{0}_{M'_1}$ and hence he does not aborts.

Same argument applies to other rounds, which implies that the protocol never aborts. Thus, the state at the end of the protocol is $$  \Tr_{CE_AT_B}(U_{r+1}U_{r}\ldots U_2U\pi_{RBCAE_AE_B}) = \omega_{RBC_0A}.$$
\end{proof}

 Thus, from equation \ref{eq:closegoodstate}, it holds that $$\P(\E(\omega_{RBCA}\otimes\theta_{E_AE_B}),\omega_{RBC_0A})\leq \sqrt{\frac{8\eps}{e_d\cdot d}}+\sqrt{\mu}.$$

Quantum communication cost of the protocol is at most $$\text{max}_{(i_1,i_2\ldots i_r)\in \G'}(\log((i_1+1)\cdot (i_2+1)\ldots (i_r+1)) \leq 2\cdot\text{max}_{(i_1,i_2\ldots i_r)\in \G'}(\log(i_1\cdot i_2\ldots i_r)\leq \frac{2C}{(1-\eps)\mu}.$$ This completes the proof.  

\end{proof}

Similarly, we have the corollary for quantum state transfer.
\begin{cor}
\label{exptoworstmerge}
Fix an error parameter $\mu>0$. There exists a $r$-round communication protocol for state transfer of $\omega'_{RC}$ with worst case quantum communication cost atmost $\frac{2C}{\mu(1-\eps)}$ and error at most $\sqrt{\frac{8\eps}{e_d\cdot d}}+\sqrt{\mu}$.
\end{cor}

Next two lemmas obtain lower bound on worst case quantum communication cost of quantum state redistribution of $\omega_{RBCA}$ and quantum state transfer of $\omega'_{RC}$.

\begin{lemma}
\label{redistworstcase}
Let $d$, the local dimension of register $B$, be such that $d>2^{18}$. Then worst case quantum communication cost of any interactive entanglement assisted quantum state redistribution protocol of the state $\omega_{RBCA}$, with error $\delta < \frac{1}{6}$, is at least $\frac{1}{6}\log(d)$.
\end{lemma}
\begin{proof}
Following lower bound on worst case quantum communication cost for interactive quantum state redistribution of the state $\omega_{RBCA}$, with error $\delta$, has been shown (\cite{Berta14}, Section $5$, Proposition $2$): $$\frac{1}{2}(\imaxdelta{R}{BC}_{\omega}-\imax{R}{B}_{\omega}).$$
Recall, from definition \ref{staterediststate}, that $\omega_{RBC}=\frac{1}{d_a}\sum_{a=1}^{d_a} \ketbra{a}_{R_A}\otimes\omega^a_{R'BC}$ is a \textit{classical-quantum} state. Consider, 
\begin{eqnarray*}
\imaxdelta{R}{BC}_{\omega} &\geq& \inf_{\rho_{RBC}\in \ball{\delta}{\omega_{RBC}}}\mutinf{R}{BC}_{\rho}\\ &\geq& \inf_{\rho_{R}\in \ball{\delta}{\omega_{R}}}S(\rho_R) + \inf_{\rho_{BC}\in \ball{\delta}{\omega_{BC}}}S(\rho'_{BC}) - \sup_{\rho_{RBC}\in \ball{\delta}{\omega_{RBC}}}S(\rho_{RBC}) \\ &\geq& \mutinf{R}{BC}_{\omega} - 3\delta\log(d) - 3 \quad (\text{Fact \ref{fact:fannes}})\\ &\geq& \frac{1}{d_a}\sum_{a}\mutinf{R'}{BC}_{\omega^a} - 3\delta\log(d)-3 \quad (\text{Fact \ref{cqmutinf}}) \\ &=& 2\log(d)-3\delta\log(d)-3.
\end{eqnarray*}

To bound $\imax{R}{B}_{\omega}$, notice that $\omega_{RB}=\frac{1}{d\cdot d_a}\sum_{a=1}^{d_a}\sum_{j=1}^d\ketbra{a}_{R_A}\otimes\ketbra{u_j}_{R'}\otimes\ketbra{v_j(a)}_{B}$ is also a \textit{classical-quantum} state. Using Fact \ref{fact:cqimax}, we obtain $\imax{R}{B}_{\omega} \leq \log(|B|) = \log(d)$.  

Thus, communication cost is lower bounded by $$\frac{1}{2}(\imaxdelta{R}{BC}_{\omega}-\imax{R}{B}_{\omega})\geq \frac{\log(d)-3\delta\log(d)-3}{2}=\frac{1-3\delta}{2}\log(d) - 1.5 > \frac{1}{6}\log(d),$$ for $d>2^{18}$.
\end{proof}

For quantum state transfer, we have following bound.

\begin{lemma}
\label{eprworstcase}
Worst case quantum communication cost for state transfer of the state $\omega'_{RC}$, with error $\delta<\frac{1}{2}$, is at least $\frac{1}{2}\log(d) + \frac{1}{2}\log(1-\delta^2)$.
\end{lemma}
\begin{proof}
The following lower bound on worst case interactive quantum communication cost of state transfer of $\omega'_{RC}$ has been shown (\cite{Berta14}, Section $5$, Proposition $2$): $$\frac{1}{2}\imaxdelta{R}{C}_{\omega'}.$$ Consider,
\begin{eqnarray*} 
\imaxdelta{R}{C}_{\omega'} &\geq& - \hmindelta{R}{C}_{\omega'} \quad (\text{Fact \ref{fact:imaxhmin}}) \\ &\geq& - \hmax{R}{C}_{\omega'}+\log(1-\delta^2) \quad (\text{Proposition 6.3, \cite{tomamichel15}}) \\ &=& \log(d) + \log(1-\delta^2)  
\end{eqnarray*}
\end{proof}

Now we proceed to proof of Theorem \ref{thm:main}.

\begin{proof}[\textbf{Proof: Theorem \ref{thm:main}}]
Suppose there exists a $r$-round communication protocol $\mathcal{P}$ for entanglement assisted quantum state redistribution of the pure state $\Psi_{RBCA}$ with error $\eps$ and expected communication cost at most $\condmutinf{R}{C}{B}_{\Psi}\cdot (\frac{1}{\eps})^p$. Then we show a contradiction for $p< 1$.
\bigskip

For a $\beta \geq 1$ to be chosen later, and $d>2^{18}$, we choose $\{e_1,e_2\ldots e_d\}$ (Definition \ref{staterediststate}) as constructed in lemma \ref{lowentropy}. Thus, $$\condmutinf{R}{C}{B}_{\Psi}\leq 2S(\Psi_C)\leq 4\frac{\log(d)}{\beta} \quad \text{(Fact \ref{informationbound})}.$$ 

Fix an error parameter $\mu$. From lemma \ref{exptoworst}, there exists a communication protocol $\mathcal{P}'$ for quantum state redistribution of $\omega_{RBCA}$, with error at most $\sqrt{\mu}+ \sqrt{8\beta\eps}$ and worst case quantum communication cost at most 

$$\frac{2\cdot\condmutinf{R}{C}{B}_{\Psi}}{\mu(1-\eps)}\cdot (\frac{1}{\eps})^p\leq 8\frac{\log(d)}{\beta\mu(1-\eps)}\cdot (\frac{1}{\eps})^p \leq 16\frac{\log(d)}{\beta\mu}\cdot (\frac{1}{\eps})^p.$$

Last inequality holds since $\eps<1/2$. Let $\beta\mu\eps^p=128$. Then $\sqrt{\mu}+ \sqrt{8\beta\eps} = \sqrt{\mu}+ \frac{32}{\sqrt{\mu}}\eps^{\frac{1-p}{2}}$, which is minimized at $\mu = 32\cdot\eps^{\frac{1-p}{2}}$. This gives $\sqrt{\mu}+ \frac{32}{\sqrt{\mu}}\eps^{\frac{1-p}{2}} = 8\sqrt{2}\cdot\eps^{\frac{1-p}{4}}$ and $\beta=4/\eps^{\frac{1+p}{2}} > 1$. 

As in the theorem, let $\eps \in [0, (\frac{1}{70})^{\frac{4}{1-p}}]$. Thus, we have a protocol for state redistribution of $\omega_{RBCA}$, with error at most $8\sqrt{2}\cdot\eps^{\frac{1-p}{4}} < \frac{1}{6}$ and worst case communication at most $\frac{1}{8}\log(d)$, in contradiction with lemma \ref{redistworstcase}.  
\end{proof}
Above argument does not hold for any $p\geq 1$ since we need to simultaneously satisfy $\beta\geq 1$, $8\beta\eps<1$ and $\mu<1$.

On similar lines, we prove Theorem \ref{thm:main2} below.

\begin{proof}[\textbf{Proof: Theorem \ref{thm:main2}}]
Suppose there exists a communication protocol for state transfer of the pure states $\tilde{\Psi}_{RC}$ with error $\eps<\frac{1}{2}$ and expected communication cost at most $S(\tilde{\Psi}_R)\cdot (\frac{1}{\eps})^p$. Then we show a contradiction for $p< 1$.

For a $\beta \geq 1$ to be chosen later, choose $a_i$ as constructed in lemma \ref{lowentropy}. Then $S(\tilde{\Psi}_R)\leq 2\frac{\log(d)}{\beta}$. 

Fix an error parameter $\mu$. From corollary \ref{exptoworstmerge}, there exists a communication protocol for state transfer of $\omega'_{RC}$, with error at most $\sqrt{\mu}+ \sqrt{8\beta\eps}$ and worst case quantum communication cost at most 

$$\frac{2S(\Psi'_R)}{\mu(1-\eps)}\cdot (\frac{1}{\eps})^p\leq \frac{4\log(d)}{\beta\mu(1-\eps)}\cdot (\frac{1}{\eps})^p \leq \frac{8\log(d)}{\beta\mu}\cdot (\frac{1}{\eps})^p.$$

Let $\beta\mu\eps^p=16$. Then $\sqrt{\mu}+ \sqrt{8\beta\eps} = \sqrt{\mu}+ \frac{8\sqrt{2}}{\sqrt{\mu}}\eps^{\frac{1-p}{2}}$, which is minimized at $\mu = 8\sqrt{2}\eps^{\frac{1-p}{2}}$. This gives $\sqrt{\mu}+ \sqrt{8\beta\eps} = \sqrt{32\sqrt{2}}\eps^{\frac{1-p}{4}}$ and $\beta=\sqrt{2}/\eps^{\frac{1+p}{2}} > 1$. 

As in the theorem, let $\eps \in [0, (\frac{1}{2})^{\frac{15}{1-p}}]$. Thus, we have a protocol for state transfer of $\omega'_{RC}$, with error at most $\sqrt{32}\eps^{\frac{1-p}{4}} < \frac{1}{2}$ and worst case communication at most $\frac{1}{2}\log(d)$, in contradiction with lemma \ref{eprworstcase}.
 
\end{proof}

\section{Conclusion}
\label{sec:conclusion}
We have shown a lower bound on expected communication cost of interactive quantum state redistribution and quantum state transfer. Main technique that we use is to construct an interactive protocol for quantum state redistribution of $\omega_{RBCA}$, using any interactive protocol for quantum state redistribution of the state $\Psi_{RBCA}$. To justify why this seems to be a necessary step, consider the sub-case of quantum state transfer. Suppose there exists a protocol for quantum state transfer of $\tilde{\Psi}_{RC}$ with expected communication cost $S(\tilde{\Psi}_R)$ and error $\eps$. We can use lemma \ref{exptoworst} to obtain an another protocol with error $\eps+\sqrt{\mu}$ and worst case communication cost at most $S(\tilde{\Psi}_R)/\mu$. But this does not lead to any contradiction, since it is straightforward to exhibit a protocol for state transfer of $\tilde{\Psi}_{RC}$ with error $\eps+\sqrt{\mu}$ and worst case communication cost $S(\tilde{\Psi}_{R})/(\eps+\sqrt{\mu})^2 < S(\tilde{\Psi}_{R})/\mu$. 

Furthermore, our argument does not apply to classical setting. This follows from the fact that we are considering a pure state $\Psi_{RBCA}$ and this allows us to obtain lemma \ref{convepr} without changing the probability distribution $p'_{i_1,i_2\ldots i_r}$ (and hence the corresponding communication weight), when we apply the map $\tilde{\E}$. 

 Some questions related to our work are as follows.
\begin{enumerate}
\item Can the bounds obtained in theorems \ref{thm:main} and \ref{thm:main2} be improved, or shown to be tight?
\item What are some applications of theorems \ref{thm:main} and \ref{thm:main2} in quantum information theory? An immediate 
application is that we obtain a lower bound on worst case communication cost of quantum state redistribution, since worst case communication cost is always larger than expected communication cost of a protocol. 
\item Is it possible to improve the direct sum result for entanglement assisted quantum information complexity obtained in \cite{Dave14}? 
\end{enumerate}

\section*{Acknowledgment} 
I thank Rahul Jain for many valuable discussions and comments on arguments in the manuscript. I also thank Penghui Yao and Venkatesh Srinivasan for helpful discussions. 

This work is supported by the Core Grants of the Center for Quantum Technologies (CQT), Singapore.

\bibliographystyle{alpha}
\bibliography{references}
\end{document}